\documentclass{victor}

\usepackage{amstext}
\usepackage{amsmath}
\usepackage{amssymb}
\usepackage{amsthm}
\usepackage{diagxy}

\theoremstyle{plain}

\newtheorem{Definition}{Definition}
\newtheorem{Theorem}{Theorem}

\newtheorem{Proposition}[Theorem]{Proposition}

\begin{document}

\title{Diophantine properties of Brownian motion: recursive aspects}
\author{Willem L Fouch\' e}%\thanks{This paper was written with support from the national science foundation. Grant 123}}
%Title and authors' names for the running head (shorten if too long, otherwise just the same as above) 
\runningtitle{Diophantine properties of Brownian motion: recursive aspects}
\runningauthor{Willem L Fouch\' e}
\author{Willem L. Fouch\'{e}\\
\it Department of Decision Sciences,\\\it School of Economic Sciences\\\it University of South
Africa, PO Box 392, 0003 Pretoria, South
Africa\\fouchwl@gmail.com}
%\date{}
% Please add your affiliation here as a comment.
% it will be printed on a separate page in the book

\altmaketitle
%\maketitle

\begin{abstract}
We use recent  results on the Fourier analysis of the zero sets of Brownian motion to  explore the diophantine properties of an algorithmically random Brownian motion ( also known as a  complex oscillation). We discuss the construction and definability of perfect sets which are linearly independent over the rationals directly from Martin-L\" of random reals. Finally we explore the recent work of Tsirelson on countable dense sets to study the diophantine properties of local minimisers of Brownian motion.
\end{abstract}

%\begin{document}
%\maketitle
\section{Introduction} 
\label{Introduction}

%\usepackage{amsmath}
%\usepackage{amsfonts}
%\usepackage{amssymb}
%$\boxplus_mK$$\biguplus_mK$ 

A Brownian motion on the unit interval is  {\it
algorithmically random} if it meets all effective (Martin-L\" of) statistical tests, now expressed in terms of
the statistical events associated with Brownian motion on the unit
interval.  The class of functions corresponds exactly, in  the language of Weihrauch \cite{wei:1,wei:2},  G\'acs \cite{gac:1} and specialised by Hoyrup and Rojas \cite{horo:1},  in the context of algorithmic randomness, to the Martin-L\" of random elements of the computable measure space 
 \[ \mathcal{R}= (C_0[0,1],d,B,W),\] where $C_0[0,1]$ is the set of the continuous functions on the unit interval that vanish at the origin, $d$ is the metric induced by the uniform norm, $B$ is the countable set of piecewise linear functions vanishing at the origin with slopes and  points of non-differentiability all rational numbers and where $W$ is the Wiener measure. We
shall also refer to such a Brownian motion as a {\it complex
oscillation}. This terminology was suggested to the author   by the following Kolmolgorov theoretic interpretation  of this notion
\cite{ap:1}: One can characterise a Brownian motion which is
generic (in the sense just stated) as an effective and uniform limit
of a sequence $(x_n)$ of ``finite random walks'', where, moreover,
each $x_n$ can be encoded by a finite binary string $s_n$ of length
$n$, such that the (prefix-free) Kolmogorov complexity, $K(s_n)$, of $s_n$
satisfies, for some constant $d > 0$, the inequality $K(s_n) > n-d$
for all values of $n$. (See Definition \ref{definition:ap}, introduced by Asarin and Prokovskii \cite{ap:1}, in Section \ref{section:co} below.)
We shall  study the  images of certain ultra-thin sets (perfect sets of Hausdorff dimension zero) under a complex oscillation. We have shown  in  \cite{fo:5} that these images are perfect sets whose elements are linearly independent over the field of rational numbers. In this paper we discuss the definability of these sets, within the recursion-theoretic hierarchy, by exploiting the recursive isomorphism constructed in \cite{fo:2} between the Kolmogorov-Chaitin random reals and the class of suitably encoded versions of  complex oscillations. We shall also utilise Tsirelson's theory of countable dense random sets \cite{tsi:1} to study the diophantine properties of the local minimisers of Brownian motion. The local minimizers of a complex oscillation is studied in \cite{fo:3}.

We  shall utilise recent results by Mukeru and the author \cite{fomu:1}  on the rate of decay of the Fourier transform of the delta function of a continuous version of Brownian motion to  identify some diophantine properties of the zero set of a complex oscillation. For more on the Fourier and consequent Diophantine properties of the sample paths of Brownian motion the reader is referred to the paper by \L aba and Pramanik \cite{LaPr:1}.  We shall also show that some of these phenomena can be expressed within the hyperaritmetical hierarchy and pose the problem as to whether this is essentially so.

 The author is very grateful to the Department of Mathematics at the Corvinus University, Budapest, for hosting my frequent visits to the department and for teaching and sharing with me so much of the subtleties of  measure theory and stochastic processes.

This  research is being partially supported by the National Research Foundation (NRF) of South Africa as well as by a Marie Curie International Research
Staff Exchange Scheme Fellowship (COMPUTAL PIRSES-GA-2011- 294962) within the 7th European Community Framework Programme.

%%%%%%%%%%%%%%%%%%%%%%%%%%%%%%%%%%%%%%%%%%%%%%%%%%%%%%%%%%%%%%%%%%%%%%%%%%%
\section{Preliminaries from Brownian motion and geometric measure theory}
 A random variable $X$ with mean $\mu$ and variance $\sigma^2$ is {\it normal} if it has a density function of the form 
\[\frac{1}{\sqrt{2\pi}\;\sigma}\;e^{-(t-\mu)^2/2\sigma^2}.\]
If $(\Omega, \mathcal{F}, P)$ is a probability space and $X$ is a real-valued random variable on $\Omega$, the measure $\mu$ on $\mathcal{F}$ given by
\[F \mapsto P(X^{-1}(F)),\; \; F \in \mathcal{F},\]
is called the {\it distribution} of $X$. %If $X_1,X_2$ are independent random variables on $\Omega$ having distributions $\mu_1,\mu_2$ respectively, the distribution of $X_1+X_2$ is given by the convolution product $\mu_1*\mu_2$ of the two distributions $\mu_1$ and $\mu_2$.

A Brownian motion on the unit interval is a real-valued function $(\omega,t) \mapsto X_\omega(t)$ on $\Omega \times [0,1]$, where $\Omega$ is the underlying space of some probability space, such that $X_\omega(0)=0$ a.s. and for $t_1 < \ldots < t_n$ in the unit interval, the random variables $X_\omega(t_1),X_\omega(t_2)-X_\omega(t_1), \cdots, X_\omega(t_n)-X_\omega(t_{n-1})$ are statistically independent and normally distributed with means all $0$ and variances $t_1,t_2-t_1,\cdots,t_n-t_{n-1}$, respectively.

It is a fundamental fact that any Brownian motion has a ``continuous version''(see, for example \cite{fre:1}). This means the following: Write $\Sigma$ for the $\sigma$-algebra of Borel sets of $C[0,1]$ where the latter is topologised by the uniform norm topology. There is a unique probability measure $W$ on $\Sigma$ such that for $0\leq t_1 < \ldots <t_n \leq 1 $ and  for a Borel subset $B$ of $\mathbb{R}^n$, we have
\[ P(\{\omega \in \Omega:(X_\omega(t_1), \cdots, X_\omega(t_n)) \in B \}) = W(A) ,\]
where
\[A =\{x \in C[0,1]: (x(t_1), \cdots, x(t_n)) \in B\}).\]
The measure $W$ is known as the {\it Wiener measure}. We shall usually write $X(t)$ instead of $X_\omega(t)$. 

%\section{Tsirelson's results on local minima of Brownian motion}
%\label{section:tsirelson}
%If $X$ is a continuous version of one-dimensional Brownian motion, denote by $\mathcal{N}(X)$ the set of local minima of $X$. It is well-known that $\mathcal{N}(X)$ is almost surely a dense and countable set and that all the local minima of $X$ are strict. This means that, for each closed subinterval $I$ of the closed unit interval, there is a unique $\nu \in I$ where the minimum of $X$ on $I$ is assumed.

%\begin{theorem}
% IIn \cite {tsi:1} Tsirelson also proved the following deep and remarkable theorem:
%f $X$ is a continuous version of Brownian motion on the unit interval and $B$ is a Borel subset of the unit interval such that $\lambda(B) > 0$, then almost surely, $B \cap \mathcal{N}(X) \neq \emptyset$. On the other hand, if $\lambda(B) = 0$, then almost surely, $B \cap \mathcal{N}(X) = \emptyset$. In particular, if $\lambda(B)=1$, then, almost surely, $\mathcal{N}(X) \subset B$.
%\label{theorem:tsirelson}
%\end{theorem}

%This theorem gives a clear expression of the extent to which the local minima of $X$ can be viewed as a ``generic'' countable dense random set. It is an interesting problem to investigate the extension of  these results to random (relative to the probability space on which $X$ has been defined) Borel sets $B$. We shall discuss some instances of this problem in  this paper. 

For a compact subset $A$ of Euclidean space $\mathbb{R}^d$ and real numbers $\alpha,\epsilon$ with $0 \leq \alpha < d$ and $\epsilon > 0$, consider all coverings of $A$ by balls $B_n$ of diameter $\leq \epsilon$ and the corresponding sums \[\sum_n|B_n|^\alpha,\]
where $|B|$ denotes the diameter of $B$. All the metric notions here are to be understood in terms of the standard $\ell^2$ norms on Euclidean space. The infimum of the sums over all coverings of $A$ by balls of diameter $\leq \epsilon$ is denoted by $H_\alpha^\epsilon(A)$. When $\epsilon$ decreases to $0$, the corresponding $H_\alpha^\epsilon(A)$ increases to a limit (which may be infinite). The limit is denoted by $H_\alpha(A)$ and is called the Hausdorff measure of $A$ in dimension $\alpha$.

If $0 < \alpha < \beta \leq d$, then, for any covering $(B_n)$ of $A$,
\[\sum_n|B_n|^\beta \leq \sup_n|B_n|^{\beta-\alpha}\sum_n|B_n|^\alpha,\]
from which it follows that
\[ H_\beta^\epsilon(A) \leq \epsilon^{\beta - \alpha}H_\alpha^\epsilon(A).\]
Hence if $H_\alpha(A) < \infty$, then $H_\beta(A)=0$. Equivalently,
\[ H_\beta(A) > 0 \Longrightarrow H_\alpha(A)=\infty.\]
Therefore,
\[ \sup\{\alpha:H_\alpha(A)=\infty\}=\inf\{\beta:H_\beta(A)=0\}.\]
This common value is called the Hausdorff dimension of $A$ and denoted by dim$_hA$.

If $\alpha$ is such that $0 < H_\alpha(A) < \infty$, then $\alpha =$ dim$_h A$. However, if $\alpha =$ dim$_h A$, we cannot say anything about the value of $H_\alpha(A)$.

 It is easy to check that $A \mapsto H_\alpha(A)$ defines an outer measure which is invariant under translations and rotations, and homogeneous of degree $\alpha$ with respect to dilations.

%In the definition we used only the structure of $A$ as a metric space. Therefore, the notions $H_\alpha(A)$ and $\mbox{dim}A$ make sense for any metric space.

If $A$ is a Borel subset of Euclidean space, the set of non-zero Radon measures with support contained in $A$ is denoted by $M^+(A)$. For a given $\mu \in M^+(A)$, the energy integral of $\mu$ with respect to the kernel $|x|^{-\alpha}$ is given by
\[I_\alpha(\mu) = \int_{\mathbb{R}^d}\int_{\mathbb{R}^d}\frac{d\mu(x)d\mu(y)}{|x-y|^\alpha}.\]
We say that $\mu$ has finite energy with respect to $|x|^{-\alpha}$ when $I_\alpha(\mu) < \infty$. If $A$ carries positive measures of finite energy with respect to $|x|^{-\alpha}$ we say that $A$ has positive capacity with respect to $|x|^{-\alpha}$ and we write 
\[\mbox{Cap}_\alpha(A) > 0.\]
If $A$ carries no positive measure of finite energy with respect to $|x|^{-\alpha}$, we say that $A$ has capacity zero with respect to this kernel and we write $\mbox{Cap}_\alpha(A)=0$.

It follows from the Fourier analysis of temperate distributions that
\begin{equation}
I_\alpha(\mu)=C(\alpha,d)\int_{\mathbb{R}^d}|\hat\mu(\xi))|^2|\xi|^\alpha\frac{d\xi}{|\xi|^d},
\label{equation:energy}
\end{equation}
when $0 < \alpha < d$, where $C(\alpha,d)$ is a positive constant and where,  moreover,
\[\hat\mu(\xi)=\int_{\mathbb{R}^d}e^{is\xi}d\mu(s),\]
is the Fourier transform of the measure $\mu$. (For a proof see \cite{mat:1}.)

 We shall make frequent use of the following  very fundamental fact.
\begin{Proposition}
For a compact subset $A$ of $\mathbb{R}^d$ and $0 < \alpha < \beta < d$,
\[H_\beta(A) > 0 \Rightarrow \mbox{\emph {Cap}}_\alpha(A)>0 \Rightarrow H_\alpha(A) > 0.\]
Hence \[\sup\{\alpha:\mbox{\emph {Cap}}_\alpha(A) > 0\} = \mbox{\emph{dim}}_h A,\]
or, equivalently,
\[\sup\{\alpha:I_\alpha(\mu) < \infty\} = \mbox{\emph{dim}}_h A.\]
\label{proposition:hausdorfcapacity}
\end{Proposition}

The Fourier dimension of a compact set is the supremum of positive real numbers $\alpha <1$ such that for some non-zero Radon measure $\mu$ supported by $E$ , it is the case that 
\[ |\hat\mu(\xi)|^2 \leq \frac{1}{|\xi|^\alpha}, \]
for $|\xi|$ sufficiently large. The Fourier dimension of $E$ is denoted by $\dim_f(E)$. Clearly, by (\ref{equation:energy}),
\[\dim_f(E) \leq \dim_h(E),\]
for all compact sets $E$. The set is called a {\it Salem} set if $\dim_f(E) = \dim_h(E)$.

The following question posed by Beurling was addressed %was posed by Beurling ??? 
and solved in the positive by Salem in 1950. (On singular monotonic functions whose spectrum has a
given Hausdorff dimension
By R. Salem (1950), Ark Mat 1,353-365.)

 Given a number $\alpha \in (0, 1)$, does there exist a closed set on the line whose Hausdorff dimension is $\alpha$  that carries a Borel  measure $\mu$ whose Fourier transform $$\hat \mu (u) = \int _{\mathbb{R}}e^{iux}d\mu(x)$$ is dominated by $|u|^{-\alpha/2}$ as $|u| \rightarrow \infty$?

It follows from (\ref{equation:energy}) that given a compact subset $E$ of $[0, 1]$ with Hausdorff dimension $\alpha \in (0, 1)$, the number $\alpha/2$ is critical for this question to have an affirmative answer. 

Salem  proved this result by constructing for every $\alpha$ in the unit interval, a random measure $\mu$ (over a convenient probability space)  whose support has Hausdorff dimension $\alpha$ and which satisfies the Beurling-requirement with probability one. 

It was recently shown by the author  in collaboration with George Davie and Safari Mukeru that such sets can also be constructed by looking at Cantor type sets $E$ with computable ratios $\xi$ and then to consider the image of $E$ under a complex oscillation.
%The mean, or expected value of a random variable $X$ will be denoted by $E(X)$. 
%Two random variables $X$ and $Y$ on possibly different probability spaces are said to be {\it similar} when they have the same probability distributions. 

The following theorem illustrates the rich diophantine structure of sets $E$ of non-zero Fourier dimension.  Even though the proof method is well-known in geometric measure theory, we give a full proof, for we need sharper estimates than waht the author could have found in the literature.
\begin{Theorem} (folklore)
Suppose $E$ is a compact subset of reals such that, for every $\epsilon >0$, there is some $\mu \in M_+(E)$ and $ 0< \alpha < 1$, such that, for some constant $C=C(\epsilon)$, it is the case that
\[|\hat\mu(\xi)|^2 \leq C|\xi|^{-\alpha + \epsilon},\]
as $|\xi| \rightarrow \infty$.
Then, if $k$ is a natural number such that $k\alpha > 1$, it will follow, upon writing 
\[E_k= E +\cdots +E\;\;(k\;\mbox {times}),\]
that
\[\mathbb{R} = \bigcup_{n < \omega} n(E_k -E_k).\]
Moreover, if $ A$ is any finite set of real numbers, then $E_k$ will contain an affine  (a translated and rescaled) copy of $A$.
\label{theorem:wlf10}
\end{Theorem}

\begin{proof} Set
\[ \nu = \mu *\cdots *\mu \;\;(k\;\mbox {times}).\] (Here $*$ denotes the convolution product.)
Clearly, by choosing $\epsilon  > 0$ such that $k(\alpha -\epsilon ) >1+\epsilon$, we have , for $|\xi|$ large,

\[|  \hat\nu(\xi)|^2 = |\hat \mu(\xi)^2|^k  \leq C^k|\xi|^{-k\alpha+k\epsilon }\leq  C^k|\xi|^{-1-\epsilon}.\]
It follows that the function $\hat\nu$ is in $L^2(\mathbb{R})$. Since $\nu$ is a non-zero measure, it follows from Parseval's theorem   that $\nu$ is absolutely continuous with respect to Lebesgue measure.  In particular, $ \mbox{supp}\;\nu$ has non-zero Lebesgue measure.  Since \[ \mbox{supp}\; \nu \subset E_k,\] we conclude that $E_k$ has non-zero Lebesgue measure.  It follows from Steinhaus's  theorem \cite{ste:1} that $E_k-E_k$ has zero as an interior point. This concludes the first part of the theorem.

The second part follows from  the following beautiful remark  \cite{LaPr:1}:  If $F$ is any set of positive Lebesgue measure, then $F$  will contain an affine copy of any finite set $A$ of real numbers. This is, as noted by \L aba and Pramanik   \cite{LaPr:1}, a consequence of Lebesgue's density theorem.
\end{proof}
\section{Complex oscillations}
\label{section:co}
The set of non-negative integers is denoted by $\omega$ and we write $\mathcal{B}$ for the Cantor space $\{0,1\}^\omega$. The set of words over the alphabet $\{0,1\}$ is denoted by $\{0,1\}^*$. If $a \in \{0,1\}^*$, we write $|a|$ for the length of $a$. If $\alpha=\alpha_0\alpha_1\ldots$ is in $\mathcal{B}$, we write $\overline{\alpha}(n)$ for the word $\prod_{j<n}\alpha_j$. We use the usual recursion-theoretic terminology $\Sigma_r^0$ and $\Pi_r^0$ for the arithmetical subsets of \(\omega^k \times \mathcal{B}^l,\;k,l \geq 0 \). (See, for example, \cite{hin:1}). We write $\lambda$ for the Lebesgue probability measure on $\mathcal{B}$. For a binary word $s$ of length $n$, say, we write $[s]$ for the ``interval'' $\{\alpha \in \mathcal{B}: \overline{\alpha}(n) = s \}$. A sequence $(a_n)$ of real numbers converges {\it effectively} to $0$ as $n \rightarrow \infty$ if for some total recursive $f:\omega \rightarrow \omega$, it is the case that $|a_n| \leq (m+1)^{-1}$ whenever $n \geq f(m)$.

For any finite binary word $a$ we denote its (prefix-free) Kolmogorov complexity by $K(a)$.  Recall that an infinite binary string $\alpha$ is Kolmogorov-Chaitin complex if 

\begin{equation}
\exists_d \forall_n\;K(\overline{\alpha}(n)) \geq n-d.
\label{eq:Kolmogorov}
\end{equation}
In the sequel, we shall denote this set by $KC$ and refer to its elements as $KC$-strings. (See, e.g., \cite{ch:2}, \cite{marlo:1} or \cite{nie:1} for more background.)

%We next survey the results from \cite{ap:1}, \cite{fo:1} and \cite{fo:2} which will play an important role in this discussion.  

For $n \geq 1$, we write $C_n$ for the class of continuous functions on the unit interval that vanish at $0$ and are linear with slopes $\pm \sqrt{n}$ on the intervals \([(i-1)/n,i/n]\;,i=1,\ldots ,n\). With every $x \in C_n$, one can associate a binary string $a =a_1\cdots a_n$ by setting $a_i=1$ or $a_i =0$ according to whether $x$ increases or decreases on the interval $[(i-1)/n,i/n]$. We call the sequence $a$ the code of $x$ and denote it by $c(x)$. The following notion was introduced by Asarin and Prokovskii in \cite{ap:1}.
\begin{Definition}

A sequence $(x_n)$ in $C[0,1]$ is \emph{ complex}  if $x_n \in C_n$ for each $n$ and there is a constant $d > 0$ such that $K(c(x_n)) \geq n-d$ for all $n$. A function $x \in C[0,1]$ is a \emph{complex oscillation} if there is a complex sequence $(x_n)$ such that $\|x-x_n\|$ converges effectively to $0$ as $n \rightarrow \infty$.\
\label{definition:ap}
\end{Definition}

The class of complex oscillations is denoted by $\mathcal{C}$. 

 In \cite{fo:2} the author constructed a bijection \(\Phi:KC \rightarrow \mathcal{C} \) which is  effective in the following sense: If $\alpha \in KC$ and $m < \omega$, one can effectively construct from the first $m$ bits of $\alpha$, a function $p_m$, where $p_m$ is a finite linear combination of piecewise linear functions with rational coefficients, such that, for some absolute positive constant $C$, the complex oscillation $\Phi(\alpha)$ is approximated by the sequence $(p_m)$ as follows:
\begin{equation}
\mbox{sup}_{t \in [0,1]}|\Phi(\alpha)(t) -p_m(t)| \leq C\log m/\sqrt{m}
\label{eq: wlf}
\end{equation}
for all $m > M_\alpha$, where $M_\alpha$ is a constant that depends on $\alpha$
only. Conversely, if $x \in \mathcal{C}$, then one can compute, relative to an
infinite binary string which encodes the values of a complex oscillation  $x$ at the rational
numbers in the unit interval, the $KC$-string $\alpha$ such that $\Phi(\alpha)
= x$.

In \cite{fo:2} the author proved:

\begin{Theorem}There is a uniform algorithm that, relative to any $KC$-string $\alpha$, with input a rational number $t$ in the unit interval and a natural number $n$, will output the first $n$ bits of the the value of the complex oscillation $\Phi(\alpha)$ at the value $t$. 
 \label{theorem:isomorphism}
\end{Theorem}

This result plays a crucial r\^ ole in this paper, for it will enable us to show how the sample path properties of  a complex oscillation  $\Phi(\alpha)$  (and hence of a typical Brownian motion) can be described within the arithmetical hierarchy relative to  the associated $KC$-string $\alpha$. In this way, as was stated in the introduction of this paper,  one finds an explicit unfolding of the incredibly  rich  geometry that is enfolded in every KC-string $\alpha$ by merely regarding such an $\alpha$ as an {\it encoding} of a complex oscillation or, equivalently, of an (effectively) generic Brownian motion.

The mapping $\Phi$ is also a measure-theoretic isomorphism in the following (standard) sense:  Write $\lambda$ for the Lebesgue measure on the space $\{0,1\}^\omega$ and write $W$ for the Wiener measure on $C[0,1]$. Then, for any Borel subset $A$ of $C[0,1]$ with the uniform norm topology, we have
\[\lambda(\Phi^{-1}(A))=W(A).\]
In other words, $W$ is the pushout of $\lambda$ under $\Phi$. We shall frequently denote $\Phi(\alpha)$ by $x_\alpha$.

We follow \cite{fo:1} to define  an analogue of a $\Pi_2^0$ subset of $C[0,1]$ which is of constructive measure $0$. If $F$ is a subset of $C[0,1]$, we denote by $\overline{F}$ its topological closure  in $C[0,1]$ with the uniform norm topology. For $\epsilon > 0$, we let $O_\epsilon(F)$ be the $\epsilon$-ball \(\{f \in C[0,1]:\exists_{g \in F}\|f-g\| < \epsilon\}\) of $f$. (Here $\|.\|$ denotes the supremum norm.) 
We write $F^0$ for the complement of $F$ and $F^1$ for $F$.
\begin{Definition}

A sequence $\mathcal{F}_0=(F_i:i < \omega)$ in $\Sigma$ is an \emph{effective generating sequence} if
\begin{enumerate}

\item for $F \in \mathcal{F}_0$, for $\epsilon > 0$ and $\delta \in \{0,1\}$, we have, for $G=O_\epsilon(F^\delta)$ or for $G=F^\delta$, that $W(\overline{G})=W(G)$,
\item there is an effective procedure that yields, for each sequence \( 0 \leq i_1 < \ldots < i_n < \omega\) and $k < \omega$ a binary rational number $\beta_k$ such that 
\[|W(F_{i_1} \cap \ldots \cap F_{i_n}) - \beta_k| < 2^{-k},\]
\item for $n,i < \omega$, a strictly positive rational number $\epsilon$ and for $x \in C_n$, both the relations $x \in O_\epsilon(F_i)$ and $x \in O_\epsilon(F_i^0)$ are recursive in $x,\epsilon,i$ and $n$, relative to an effective representation of the rationals.

\end{enumerate}
\label{def:effective}

\end{Definition}

If $\mathcal{F}_0 = (F_i:i < \omega)$ is an effective generating sequence and $\mathcal{F}$ is the Boolean algebra generated by $\mathcal{F}_0$, then there is an enumeration $(T_i:i < \omega)$ of the elements of $\mathcal{F}$ (with possible repetition) in such a way, for a given $i$, one can effectively describe $T_i$ as a finite union of sets of the form
\[F = F_{i_1}^{\delta_1} \cap \ldots \cap F_{i_n}^{\delta_n}\]
where $0 \leq i_1 < \ldots < i_n$ and $\delta_i \in \{0,1\}$ for each $i \leq n$. We call any such sequence $(T_i:i < \omega)$ a {\it recursive enumeration} of $\mathcal{F}$. We say in this case that $\mathcal{F}$ is {\it effectively generated} by $\mathcal{F}_0$ and refer to $\mathcal{F}$ as an {\it effectively generated algebra} of sets.

Let $\left(T_{i} : i <\omega\right)$ be a recursive enumeration of the algebra $\mathcal{F}$ which is effectively generated by the sequence $\mathcal{F}_{0}=\left(F_{i}:i<\omega\right)$ in $\Sigma$. It is shown in \cite{fo:1} that there is an effective procedure that yields, for $i,k<\omega$, a binary rational $\beta_{k}$ such that

$$|W\left(T_{i}\right)-\beta_{k}|<2^{-k},$$
in other words, the function $i \mapsto W(T_i)$ is computable.

A sequence $(A_n)$ of sets in $\mathcal{F}$ is said to be $\mathcal{F}$-{\it semirecursive} if it is of the form $(T_{\phi(n)})$ for some total recursive function $\phi: \omega \rightarrow \omega$ and some effective enumeration $(T_i)$ of $\mathcal{F}$. (Note that the sequence $(A_n^c)$, where $A_n^c$ is the complement of $A_n$, is also an $\mathcal{F}$-semirecursive sequence.) In this case, we call the union $\cup_nA_n$ a $\Sigma_1^0(\mathcal{F})$ set. A set is a $\Pi_1^0(\mathcal{F})$-set if it is the complement of a $\Sigma_1^0(\mathcal{F})$-set. It is of the form $\cap_n A_n$ for some $\mathcal{F}$-semirecursive sequence $(A_n)$. A sequence $(B_n)$ in $\mathcal{F}$ is a {\it uniform} sequence of $\Sigma_1^0(\mathcal{F})$- sets if, for some total recursive function $\phi:\omega^2 \rightarrow \omega$ and some effective enumeration $(T_i)$ of $\mathcal{F}$, each $B_n$ is of the form 
\[ B_n = \bigcup_m T_{\phi(n,m)}.\]
In this case, we call the intersection $\cap_nB_n$ a $\Pi_2^0(\mathcal{F})$-set. If, moreover, the  Wiener-measure of $B_n$ converges {\it effectively} to $0$ as $n \rightarrow \infty$, we say that the set given by $\cap_nB_n$ is a $\Pi_2^0(\mathcal{F})$-set of constructive measure $0$. 

The proof of the following theorem appears in \cite{fo:1}.
\begin{Theorem}

Let $\mathcal{F}$ be an effectively generated algebra of sets. If $x$ is a complex oscillation, then $x$ is in the complement of every $\Pi_2^0(\mathcal{F})$-set of constructive measure $0$.
\label{theorem:quoteone}
\end{Theorem}
This means, that every complex oscillation is, in an obvious sense, $\mathcal{F}$-Martin-L\" of random.
\begin{Definition} 

An effectively generated algebra of sets $\mathcal{F}$   is \emph{universal} if the class $\mathcal{C}$ of complex oscillations is definable  by  a   single $\Sigma_2^0(\mathcal{F})$-set,   the complement of which is a set of constructive measure $0$.
In other words, $\mathcal{F}$ is universal iff a continuous function $x$  on the unit interval is a complex oscillation iff $x$ is $\mathcal{F}$-Martin-L\" of random.
\label{definition:universal}
\end{Definition} 

We introduce  two classes  of effectively generated algebras $\mathcal{G}$ and {$\mathcal{M}$ which are very useful for reflecting properties of one-dimensional Brownian motion into complex oscillations. %As was shown in  \cite{fo:1} and \cite{fo:2}, they are both universal in the sense just stated. 

Let $\mathcal{G}_0$ be a family of sets in $\Sigma$ each having a description of the form:
\begin{equation}
 a_1X(t_1) + \cdots + a_nX(t_n) \leq L
\label{eq:gaussone}
\end{equation}
or of the form (\ref{eq:gaussone}) with $\leq$ replaced by $<$, where all the $a_j,t_j\;(0 \leq t_j \leq 1)$ are non-zero rational numbers, $ L$ is a recursive real number and $X$ is one-dimensional Brownian motion. 

%If  $\epsilon > 0$ and $G \in \Sigma$ is described by (\ref{eq:gaussone}), we have that \( O_\epsilon(G) \) is described by the inequality
%\begin{equation}
% a_1X(t_1) + \cdots +a_nX(t_n) < L + \epsilon\sum_j |a_j| 
%\label{equation:newgaussone}
%\end{equation}
%while $O_\epsilon(G^0)$ is given by
%\begin{equation}
 %a_1X(t_1) + \cdots +a_nX(t_n) > L - \epsilon \sum_j|a_j|. 
%\label{equation:newgausstwo}
%\end{equation}
We require that it be possible to find an enumeration $(G_i:i < \omega)$ of $\mathcal{G}_0$ such that, for given $i$, if $G_i$ is given by (\ref{eq:gaussone}), we can effectively compute the sign, the denominators and numerators of the rational numbers $a_j,t_j$ and, moreover, that the recursive real $L$  can be computed up to arbitrary accuracy. 

%This has the implication that there is an effective procedure, $\Pi$, such that, for given $i, \epsilon, m$ with $i,m < \omega$ and $\epsilon$ a positive rational, the validity of (\ref{equation:newgaussone}) and (\ref{equation:newgausstwo}) can be decided by $\Pi$ when $G_i$ is given by (\ref{eq:gaussone}) and when $X \in C_m$. 

It is shown in \cite{fo:2} that $\mathcal{G}_0=(G_i:i<\omega)$ is an effective generating sequence in the sense of Definition \ref{def:effective}. The associated effectively generated algebra of sets $\mathcal{G}$ will be referred to as a {\it gaussian algebra}. 

It is shown in \cite{fo:1} that if $\mathcal{G}_0 $ is defined by  events of the form  (\ref{eq:gaussone}) with $n=1$ and $a_1=1$,  then the associated $\mathcal {G}$ is in fact universal in the sense of Definition \ref{definition:universal}.

We shall also make frequent use of the following result from \cite{fo:1} which is an easy  consequence of Theorem \ref{theorem:quoteone}. It is the analogue, for continuous functions, of the well-known fact that Kurtz-random reals are in fact Martin-L\" of random.
\begin{Theorem}
If $B$ is a $\Sigma_1^0(\mathcal{F})$ set and $W(B) = 1 $, then $\mathcal{C}$, the set of complex oscillations, is contained in $B$.
\label{theorem:quotetwo}
\end{Theorem}

\section{Diophantine properties of zero sets of Brownian motion and complex oscillations}

The following result is proven in \cite{fomu:1}.
\begin{Theorem}
{\bf (Fouch\'e and Mukeru)}(2013). Let $X$ be a continuous version of one-dimensional Brownian motion on the unit interval. Then, almost surely, there exists a nonzero Radon measure $\mu$ with support on $Z_X$, the zero set of $X$,  such that its Fourier transform $\hat \mu$ satisfies the inequality
\begin{equation}
 \vert \hat \mu(\xi) \vert^2 \ll _\epsilon \vert \xi \vert^{-\frac{1}{2}+\epsilon},
\label{equation:FM}
\end{equation}
as $\vert \xi\vert \rightarrow \infty$. In particular, the zero-set of Brownian motion is a Salem set.
\end{Theorem}
 It would be interesting to study the existence of arithmetic propgressions in the zero sets of $X$. This question is related to the results obtained in Section 8 of \cite{LaPr:1} by \L aba en Pramanik.

By Theorem \ref{theorem:wlf10}, the preceding theorem has the following consequence:
\begin{Theorem}
For a continuous version $X$ of Brownian motion over the unit interval, we have, almost surely,
\[ \mathbb{R} = \bigcup_{n =1}^\infty n(Y_X-Y_X),\]
where
\[Y_X =  Z_X + Z_X + Z_X,\]
and $Z_X$ is the zero set of $X$. Moreover, almost surely,  for any finite set $A$ of real numbers, the set $Y_X$ will contain an affine (rescaled and translated) copy of $A$.
\end{Theorem}

We now investigate the extent to which this result can be reflected in every complex oscillation.
For a fixed $r \in \mathbb{R}$ define the subset $\Omega_r$ of $C[0,1]$ by:
\[X \in \Omega_r \leftrightarrow \exists_n \exists_{z_1,\ldots,z_6 \in Z_X}\; \big[r=n((z_1+z_2+z_3)-(z_4+z_5+z_6))\big].\]
It follows from the preceding that each $\Omega_r$ has Wiener measure one.
For a real $r$ and a natural number $\ell$, let $I_{r,\ell}$ be any interval of length $\leq \frac{1}{\ell}$ with rational endpoints which contains $r$.
%1 We also construct the intervals in such a way, that, if $r$ is nonzero, the distance from the intervals to $0$ is uniformly bounded from below..

For a real $r$, a continuous function $X$ on the unit interval and an natural number $\ell$ define the predicate $P(r,\ell,X)$ by:
\[ P(r,\ell,X) \leftrightarrow \exists_n \exists_{t_1, \ldots, t_6 \in [0,1] \cap \mathbb{Q}}\; \big[n((t_1+t_2+t_3)-(t_4+t_5+t_6)) \in I_{r, \ell}\big] \wedge \forall_{1 \leq i \leq 6} |X(t_i)| < \frac{1}{\ell}.\]
Note that for fixed $r$ and $\ell$ the predicate $P(r,\ell,X)$ is $\Sigma_1^0(\mathcal{G})$ for some (fixed) gaussian algebra $\mathcal{G}$. 

Our next aim is to show, for nonzero $r$:
\[ X \in \Omega_r \rightarrow \forall_\ell P(r,\ell,X).\]
This will have the implication that for fixed $r,\ell$, the predicate $P$ defines a $\Sigma_1^0(\mathcal{G})$--set of Wiener measure one so that in particular $P(r,\ell,x)$ will also hold for each complex oscillation $x$.

For $X \in \Omega_r$ and $\ell \geq 1$ let $n$ be a natural number and $z_1, \ldots,z_6$ be zeroes of $X$ such that $r=n((z_1+z_2+z_3)-(z_4+z_5+z_6))$. Next choose $t_1, \dots, t_6 \in [0,1] \cap \mathbb{Q}$ sufficiently close to $z_1, \ldots,z_6$ to ensure that both $|n((t_1+t_2+t_3)-(t_4+t_5+t_6))-r|<\frac{1}{\ell}$ and $|X(t_i)| < \frac{1}{\ell}$ for $i=1,\ldots,6$ holds. Consequently, we can deduce $P(r, \ell,X)$ for all $\ell$.

We have proven
\begin{Theorem}
If $x$ is a complex oscillation and $r$ is a real number then
 \[ \forall_\ell\exists_n \exists_{t_1, \ldots, t_6 \in [0,1] \cap \mathbb{Q}}\; \big[|n((t_1+t_2+t_3)-(t_4+t_5+t_6))-r|< \frac{1}{\ell}\big ] \wedge \forall_{1 \leq i \leq 6} |x(t_i)| < \frac{1}{\ell}.\]
\label{theorem:final}
\end{Theorem}

Denote the predicate in Theorem  \ref{theorem:final} by $P(x,r)$. It follows that the set defined $B$ by
\[x \in B \leftrightarrow \forall_r P(x,r) \]
contains all the complex oscillations. Define $Q(x,r)$ as $P(x,r)$ but with the first two quantifiers interchanged. Then 
\[ \mathbb{R} = \bigcup_{n < \omega} n(Y_x-Y_x) \leftrightarrow \forall_r Q(x,r).\]
It is an open problem whether the predicate $\forall_rQ(x,r)$ defines a set that will contain all complex oscillations.

\section{Hamel sets generated by complex oscillations}
\label{section:hamel}
For the historical background to and a Fourier-analytical perspective on the results of this section, the reader is referred to Chapter 5 of the book by Rudin \cite{rud:1}.

A perfect subset of the unit interval is called a {\em Hamel set}, if its elements are linearly independent over the field of rational numbers, or, equivalently, if it is a perfect subset of some Hamel basis of the reals over the rationals. Our aim is to show how Hamel sets can be generated by complex oscillations. Our results are inspired by the arguments on pp 255-257 of Kahane \cite{kah:1}.

Set
\begin{equation}
E= \left\lbrace\frac{1}{2}+\sum_{k=2}^\infty\epsilon_k\frac{1}{2^{k^2}}:\epsilon_k \in\{-1,1\}\;\mbox{for}\;\mbox{all}\;k \right\rbrace.
\label{eq:thinset}
\end{equation}
In \cite{fo:5}, the author proved:
\begin{Theorem}

If $x$ is a complex oscillation then the elements of the image $x(E)$ of the set $E$ under $x$ will be linearly independent over the field of rational numbers.
\label{th:Hamel}
\end{Theorem}

Our next aim is to show how one can use this theorem together with Theorem \ref {theorem:isomorphism} to find definitions of Hamel sets within the arithmetical hierarchy.
For $\ell \geq 2$, set
\[D_\ell= \left\lbrace\frac{1}{2}+\sum_{k=2}^\ell \epsilon_k\frac{1}{2^{k^2}}:\epsilon_k \in\{-1,1\}\;\mbox{for}\;\mbox{all}\;k=2, \ldots,\ell \right\rbrace.\]
Write $D=\cup_{\ell \geq 2} D_\ell$. Note that the topological closure of $D$ is $D \cup E$ 
We begin by proving
\begin{Proposition}

If $\alpha \in KC$, then 

\begin{equation}
 z \in x_\alpha(E) \leftrightarrow \forall_m \exists_{n>m} \exists_{t \in D_n} |x_\alpha(t)-z| < \frac{1}{2^n}.
\label{equation:predicateone}
\end{equation}
\end{Proposition}

Proof: Suppose $z = x_\alpha(t)$ where $\alpha \in KC$ and  $t \in E$ is given by
\[t = \frac{1}{2} + \sum_{k=2}^ \infty \epsilon_k \frac{1}{2^{k^2}}.\]
For $n \geq 1$ set
\[t_n = \frac{1}{2} + \sum_{k=2}^ n\epsilon_k \frac{1}{2^{k^2}}.\]
It follows from Proposition 1 in \cite{fo:2} that for some constant  $C >1$ and $n$ sufficiently large it is the case 
\[ |x_\alpha(t_n)-z| \leq C|t_n-t|^{\frac{1}{2}}\log\frac{1}{|t_n-t|}.\]
Since $|t_n-t| \leq \frac{1}{2^{n^2}}$ we conlude that for all n sufficiently large
\[|x_\alpha(t_n)-z| < \frac{1}{2^n}.\]
Conversely, suppose that $z, \alpha$  satisfy the predicate on the right-hand side of  (\ref{equation:predicateone}). With each $m$, we associate an $n=n(m) >m$ such that
$ |x_\alpha(t_n)-z| < \frac{1}{2^n}$ for some $t_n \in D_n$. The sequence $(t(m))$ has some convergent sequence with a limit $\tau$ say. Clearly $\tau \in E$, and by the continuity of $x_\alpha$, we can conclude that $x_\alpha(\tau)=z$. This concludes the proof of the Proposition.

Note that 
\[ |x_\alpha(t)-z| < \frac{1}{2^n} \leftrightarrow \exists_k |\overline{x_\alpha(t)}(k) - \overline{z(k)}| < \frac{1}{2^n} - \frac{2}{2^k}, \]
the right-hand side being  $\Sigma_1^0$ in $\alpha,z,t$ and $n$.  Consequently
\begin{Theorem}
There is a $\Pi_2^0$-formula $Q(\alpha,z)$ defined over $KC \times \{0,1\}^\omega$ such that
\[ z \in x_\alpha(E) \leftrightarrow Q(\alpha,z).\]
\end{Theorem}
Let $\Omega$ be any $\Delta_2^0$-element of $KC$ (a Chaitin real). Then $Q(\lambda i\Omega(i), z)$ is a $\Pi_3^0$-predicate in $z$ that defines a Hamel set.
We have proven
\begin{Theorem}
 There is a $\Pi_3^0$-predicate $R(z)$ over $\{0,1\}^\omega$ and a Hamel set $K$ such that
\[R(z) \leftrightarrow z \in K\].
\end{Theorem}

\section{Further developments and an open problem}
We write $S_\infty$ for the symmetric group of a countable set . We place on $S_\infty$ the pointwise
convergence  topology thus giving  $S_\infty$ the subspace topology under its embedding into the Baire space $\mathbb{N}^\mathbb{N}$. The group $S_\infty$ acts naturally (and continuously) on $(0,1)^\infty_{\neq}$:
$$\sigma.(u_j: j \geq 1) := (u_{\sigma^{-1}(j)}: j \geq 1),$$
for all $(u_j) \in (0,1)^\infty_{\neq}$ and $\sigma \in S_\infty$ . The orbit space under this action is denoted by $(0,1)^\infty_{\neq}/S_\infty$. The Borel structure on this space is given by the topology induced by the canonical mapping 
$$\pi:(0,1)^\infty_{\neq} \longrightarrow (0,1)^\infty_{\neq}/S_\infty.$$
If $X$ is a continuous function on the unit interval, then a {\it local minimizer} of $X$ is a point $t$ such that there is some closed interval $I \subset[0,1]$ containing $t$ such that the function $X$ assumes a minimum value on $I$ at the point $t$. We denote by $MIN(X)$ the set of local minimizers of $X$.

It is well-known that if $X$ is a continuous version of Brownian motion on the unit interval, then $MIN(X)$ is almost surely a dense and countable set and that all the local minimizers of $X$ are {\it strict}. This means that, for each closed subinterval $I$ of the closed unit interval, there is a unique $\nu \in I$ where the minimum of $X$ on $I$ is assumed.

 This has the implication that there is a subset $\Omega_0$ of $C[0,1]$ of full Wiener measure such that one can define a measurable mapping  $min: C[0,1] \supset \Omega_0  \longrightarrow (0,1)^\infty_{\neq} $ in such a way that the composition of $min$ with the projection $\pi$ will define a mapping $X \mapsto MIN(X)$. In the sequel this strongly random set will be denoted by $MIN$. To summarise, we have the following commutative diagram:

\[\bfig
\Vtriangle/>`>`>/[C\hbox{[}0,1\hbox{]} \supset \Omega_0`(0,1)^\infty_{\neq}.`(0,1)^\infty_{\neq}/S_\infty;\min`\hbox{MIN}`\pi]
\efig\]
Let $(m_k)$ be any random enumeration of the local minimizers of a continuous version of Brownian motion in the unit interval .

Let $q >2$ and for $k \geq 1$ set 
$$s_k = (1+m_k).$$

\begin{Theorem}

The sequence $(s_k)$ is linearly independent over $\mathbb{Q}$.
\label{theorem:minimizers}
\end{Theorem}

Proof: Let $\Omega$ be a standard Borel space. A {\it strongly } countable set in the unit interval is a measurable mapping  $X: \Omega \rightarrow (0,1)^\infty_{\neq}/S_\infty$ that factors through some (traditional) random sequence $Y$ as shown:
\[\bfig
\Vtriangle/>`>`>/[\Omega`(0,1)^\infty_{\neq}.`(0,1)^\infty_{\neq}/S_\infty;Y`X`\pi]
\efig\]
One can think of $X$ as a random countable {\it set} induced via $S_\infty$-equivalence, by a random {\it sequence} $Y$, both in the unit interval. Denote the Borel space $(0,1)^\infty_{\neq}/S_\infty$ by $CS(0,1)$.

For standard measure spaces $(\Omega_1,P_1)$ and $(\Omega_2,P_2)$, let  there be   some $P_i$-measurable strongly random variable $X_i:\Omega _i\rightarrow CS(0,1)$  such that the induced probability distributions on $CS(0,1)$ are the same.  %Then there is a probability distribution $\mathbb{P}$ on $\Omega_1 \times \Omega_2$ such that the marginal of $\mathbb{P}$ to $\Omega_i$ is $P_i$, and moreover,  for $\mathbb{P}$ almost all $(\omega_1,\omega_2) \in \Omega_1 \times \Omega_2$  it is the case that $$X_1(\omega_1)=X_2(\omega_2).$$
\newline
We say in this case that the strongly random sets $X_1$ and $X_2$ are {\it statistically similar} relative to the probabilities $P_1,P_2$ and we  write  $X_1 \sim X_2$. This means exactly that
$$P_1(X_1^{-1}(\Sigma))=P_2(X_2^{-1}(\Sigma)),$$
For all Borel subsets $\Sigma$ of $CS(0,1)$.

Write $\lambda^\infty$ for the product measure on $(0,1)^\infty$ which is the countable product of the Lebesgue measure $\lambda$ on the unit interval and write $\Lambda$ for the measure on $CS(0,1)$ 
which is the pushout of $\lambda^\infty$ under $\pi$. In other words, for a Borel subset $\Sigma$ of $CS(0,1)$,
$$\Lambda(\Sigma) =\lambda^\infty(\pi^{-1}\Sigma).$$

Write $U:(0,1)^\infty \rightarrow CS(0,1)$ for the strictly random set as defined by the following commutative diagram:
%\rightarrow (0,1)^\infty/S_\infty$ that factors via some (traditional) random sequence $Y$ as shown:
\[\bfig
\Vtriangle/>`>`>/[(0,1)^\infty_{\neq}`(0,1)^\infty_{\neq}.`CS(0,1)=(0,1)^\infty_{\neq}/S_\infty;Id`U`\pi]
\efig\]

 In statistics $U$ is a model of an unordered uniform infinite sample. Moreover, it follows from the Hewitt-Savage theorem, that for every Borel subset $\Sigma$ of $CS(0,1)$, it is the case that 
\begin{equation}
\Lambda(\Sigma) \in \{0,1\}.
\label{equation:Hewitt-Savage}
\end{equation}
Note that $\Lambda$ is non-atomic.

In \cite{tsi:1} Tsirelson proved  the truly remarkable result that
\begin{equation}
MIN \sim U.
\end{equation}  

The theorem with the uniform sequence $(u_k)$ replacing the local minimizers $(m_k)$ is known to be true. (See pp 256-260  in Meyer \cite{meyer:1}.) The set $E$ remains invariant under permutations of the indices $k$. Hence the theorem  follows from the statistical similarity of $MIN$ and $U$.

{\bf Open problem}. In \cite{fo:3} the author showed how the local minimizers of a complex oscillation $\Phi(\alpha)$ can be computed from a KC-string $\alpha$. This opens the possibility of finding analogues of Theorem \ref{theorem:minimizers} for complex oscillations.

Let us call a continuous function $x$ on the unit interval strongly random if it belongs to every $\Sigma_2^0(\mathcal{G})$ set of Wiener measure one, for some gaussian algebra $\mathcal{G}$. The set of strongly random functions is a subclass of the complex oscillations. By using the constructions in \cite{fo:3}, it can be shown that thye $s_k$ associated with a srongly random function will be linearly independent over the rationals. Whether this result can be extended to complex oscillations, is an open problem.
%%%%%%%%%%%%%%%%%%%%%%%%%%%%%%%%%%%%%%%%%%%%%%%%%%%%%%%%%%%%%%%%%%%%%%%%%%%%%%%

\end{document}